\documentclass[12pt,a4paper]{amsart}

\usepackage{latexsym,amssymb,amsmath,amsthm,amsfonts,enumerate,verbatim,xspace,
exscale}
\usepackage{graphicx}
\usepackage{color,amsbsy}

\usepackage[all]{xy}

\usepackage{hyperref}
\definecolor{darkblue}{rgb}{0,0,.5}
\hypersetup{colorlinks=true, breaklinks=true, linkcolor=darkblue, menucolor=darkblue,  urlcolor=darkblue,citecolor=darkblue}

\usepackage{setspace}

\theoremstyle{plain}
\newtheorem{theorem}{Theorem}[section]

\newtheorem{proposition}[theorem]{Proposition}

\theoremstyle{definition}

\newtheorem{remark}[theorem]{Remark}

\def\R{\mathbb{R}}

\def\F{\mathcal{F}}

\newcommand{\longto}{\longrightarrow}

\def\vv<#1>{\langle#1\rangle}
\def\ww<#1>{\langle\langle#1\rangle\rangle}

\newcommand{\dd}[2]{\mbox{$\frac{\partial #2}{\partial #1}$}}

\providecommand{\del}{\partial}

\newcommand{\Om}{\Omega}

\newcommand{\eps}{\varepsilon}

\newcommand{\by}[2]{\mbox{$\frac{#1}{#2}$}}
\newcommand{\cinf}{\mbox{$C^{\infty}$}}
\providecommand{\set}[1]{\mbox{$\{#1\}$}}

\newcommand{\gu}{\mathfrak{g}}

\newcommand{\ad}{\mbox{$\text{\upshape{ad}}$}}

\newcommand{\X}{\mbox{$\mathcal{X}$}}
\newcommand{\ds}{\mbox{$\,\delta_t$}}

\newcommand{\revise}[1]{#1}

\title[HIPS for compressible flow]{%
A Hamiltonian interacting particle system for compressible flow
}

\usepackage[foot]{amsaddr} 
\author{Simon Hochgerner}
\address{Finanzmarktaufsicht (FMA),
Otto-Wagner Platz 5, A-1090 Vienna
}
\email{simon.hochgerner@fma.gv.at} 

\begin{document}

\begin{abstract}
The decomposition of the energy of a compressible fluid parcel into slow (deterministic) and fast (stochastic) components is interpreted as a stochastic Hamiltonian interacting particle system (HIPS). It is shown that the McKean-Vlasov equation associated to the mean field limit yields the barotropic Navier-Stokes equation with density dependent viscosity. Capillary forces can also be treated by this approach. Due to the Hamiltonian structure the mean field system satisfies a Kelvin circulation theorem along stochastic Lagrangian paths. 
\end{abstract}

\maketitle

\section{Introduction} 
\subsection{The barotropic Navier-Stokes equations}
Consider a compressible barotropic fluid in an $n$-dimensional domain with periodic boundary conditions. The velocity, $u=u(t,x)$, and density, $\rho=\rho(t,x)$, are a time-dependent vector field and function, respectively, defined on the torus $M=\R^n/\mathbb{Z}^n$. 

The compressible Navier-Stokes equations with density dependent viscosity and capillary forces are
\begin{align}
\label{intro:bNS1}
    \dot{u}
    &=   
    - \nabla_{u}u
    - \rho^{-1}\nabla p
    + \rho^{-1} \Big(\textup{div}\,S  + \textup{div}\,C\Big) \\
\label{intro:bNS2}
    \dot{\rho}
    &= 
    -\textup{div}(\rho u)
\end{align}
where $\nabla_u u = \vv<u,\nabla>u = \sum u^j\del_j u$.  
The hydrostatic pressure, $p$, is assumed to be given in terms of the density, that is $p = \rho^2\mathcal{U}'(\rho)$ for a known function $\mathcal{U}$ which models the specific internal energy when the fluid is in equilibrium. 
Further, $S$ is the stress tensor, defined by 
\begin{equation}
    \label{intro:stress}
    S_{ij}
    = \nu{\rho}\Big(\del_i u^j + \del_j u^i\Big)
\end{equation}
where $\nu\ge0$ is the viscosity coefficient.  The corresponding force is   
\begin{align}
    \textup{div}\,S
    = \sum \del_i S_{ij} e_j 
    = 
    \nu (\nabla^{\top}u)\nabla{\rho} 
    + \nu {\rho}\nabla\textup{div}(u)
    + \nu \nabla_{ \nabla{\rho} }u
    + \nu{\rho}\Delta u.
\end{align}
Let $\kappa\ge0$ be a constant. The capillary tensor, $C$, is defined as 
\begin{equation}
    \label{intro:Cap}
    C 
    = \kappa\Big(
        \Big( {\rho}\Delta{\rho}
            +\by{1}{2}\vv<\nabla{\rho},\nabla{\rho}> \Big)\mathbb{I}
        - \nabla{\rho}\otimes\nabla{\rho} \Big)
\end{equation}
and satisfies $\textup{div}\,C = \kappa{\rho}\nabla\Delta{\rho}$. 

For background regarding the barotropic Navier-Stokes equations with viscosities which depend linearly on the density we refer to \cite{MV07} and references therein.

The capillary tensor~\eqref{intro:Cap} appears in this form also in \cite[Equ.~(4)]{AMW97} and goes back to Korteweg~\cite{K01}. Analytic aspects of the barotropic system with capillary forces \eqref{intro:bNS1}-\eqref{intro:bNS2} are treated in \cite{BDL03}.
Navier-Stokes equations with more general third-order spatial derivative terms are discussed in \cite{J11}.

\subsection{Description of results}
This paper is concerned with a mean field representation of solutions to the compressible Navier-Stokes system \eqref{intro:bNS1}-\eqref{intro:bNS2}, and this mean field is derived from a stochastic Hamiltonian interacting particle system (\revise{Hamiltomian IPS or} HIPS).
The HIPS picture follows from a combination of a many particle approach to fluid dynamics and a decomposition of the energy into slow (deterministic) and fast (stochastic) components. 
\revise{
The interpretation as a slow-fast decomposition is consistent with the multi-time formulation of \cite{CGH17}, where it is shown that advection along stochastic transport fields in the Eulerian representation of stochastic fluid dynamics can be obtained by homogenization.  
}

To describe the many particle approach, consider a tiny Eulerian volume (fixed in space and time) $\Delta V$ which is divided into a very large number $N$ of equal subvolumes. It is assumed that the continuum hypothesis holds in each of the infinitesimal subvolumes $\Delta V^{\alpha}$. Thus there is a mass density $\rho^{\alpha}$ for each grid index $\alpha$. Since the subvolumes are equal, so are the initial conditions for $\rho^{\alpha}$, and $\rho^{\alpha}|_0 = \rho|_0$ which is the initial condition for the overall mass density in $\Delta V$.   

Now, the fluid parcels in all of the subvolumes interact because the deterministic component of the energy of the blob of fluid in $\Delta V$ depends on the total momentum and the overall mass density. (`Momentum' shall always refer to momentum per unit volume, i.e.\ its dimension is density times velocity.) On the other hand, molecular diffusion is modeled as a set of $N$ independent (multidimensional) Brownian motions such that all $N$ individual parcels undergo their own stochastic process. Since molecules are incompressible these processes are set up as stochastic perturbations along divergence free vector fields. 

Hence the fluid parcel in $\Delta V$ consists of $N$ identical subparcels, and each subparcel follows the flow of the ensemble of subparcels (is dragged along or advected) but also undergoes its own diffusion. Heuristically, this means that the interaction in the IPS is due to the deterministic part of the motion where each infinitesimal subparcel follows the mean flow of the ensemble.   

The corresponding total Hamiltonian \eqref{e:H^N} for the fluid in $M$ is then obtained by integrating the energies over the infinitesimal subvolumes $\Delta V^{\alpha}$ and summing over all indices $\alpha$. 
Since this Hamiltonian describes a system consisting of $N$ subparcels, it is a function 
\begin{equation*}
    H^N:
    \Big(T^*(\textup{Diff}(M)\circledS\mathcal{F}(M)\Big)^N
    \to\R
\end{equation*}
which is the $N$-fold product of the phase space of compressible fluid mechanics. (The semi-direct product notation is explained in Section~\ref{sec:SD}.)
The canonical symplectic structure on the phase space then yields a Hamiltonian system of Stratonovich SDEs by adapting the construction of \cite{LCO08} to the infinite dimensional setting. This system of $N$ interacting SDEs will be called the `HIPS equations of motion'.  
In Section~\ref{sec:3A} it is explained how the Hamiltonian $H^N$ is a sum of terms involving: translational kinetic energy of the particle ensemble, equilibrium internal energy associated to the hydrostatic pressure, equilibrium internal energy associated to capillary forces, non-equilibrium internal energy due to expansion/compression along the flow, and stochastic energy associated to molecular bombardment.     
The HIPS equations of motion are derived in Section~\ref{sec:3b}, see \eqref{e:hips3}-\eqref{e:hips5}. Under the assumption that the mean field limit exists, as $N\to\infty$, the stochastic mean field equations are obtained in Section~\ref{sec:3c}. 

The mean field SDEs \eqref{e:Hmf1}-\eqref{e:Hmf2} represent the Eulerian description of the motion of a fluid parcel associated to a subdivision $\Delta V^{\alpha}$ for very large $N$. The expected fluid flow is then obtained by averaging over momenta and mass densities of all the smaller fluid parcels. But these are just the mean fields. It therefore remains to calculate the equations of motions for the latter. These equations are deterministic and given in \eqref{e:exp1}-\eqref{e:exp2}. 

In Section~\ref{sec:3d} it is shown that, if the stochastic perturbation corresponds to a Brownian motion for each Fourier mode (i.e., is a cylindrical Wiener process in the space of solenoidal vector fields), then the momentum and mass density mean fields solve the compressible Navier-Stokes system~\eqref{intro:bNS1}-\eqref{intro:bNS2}. 
This is the content of Theorem~\ref{thm:cNS}. Moreover, since the HIPS is invariant under the particle relabeling symmetry (invariance with respect to the group of diffeomorphisms) Noether's Theorem implies that a Kelvin Circulation Theorem holds along stochastic Lagrangian paths (Proposition~\ref{prop:KCT}). 

The density dependence of the viscosity, as manifest by the $\rho$-factor in \eqref{intro:stress}, is a consequence of the form of the Hamiltonian~\eqref{e:H^N}. The compressible Navier-Stokes equations are often expressed with respect to a density-independent viscosity, but it is not clear how to realize this independence in the HIPS framework. 

Mean field representations of solutions to the incompressible Navier-Stokes equation have been previously obtained by \cite{CI05,I06b} by considering a Weber functional along stochastic Lagrangian paths. In \cite{H17,H18} the above described Eulerian (H)IPS formulation has been used to also obtain a stochastic mean field representation for solutions to the incompressible Navier-Stokes equation. 

To the best of my knowledge, Theorem~\ref{thm:cNS} is the first stochastic representation of solutions to the compressible Navier-Stokes equations.  

\revise{Related approaches to fluid dynamics from the perspective of stochastic variational principles include \cite{Y83,CC07,C}, which characterize solutions to the incompressible Navier-Stokes equations by variational principles for stochastic Lagrangian paths, and generalizations in \cite{ACC14,CCR15}.
}

\subsection{Applications: NatCat models and Solvency~II capital requirements}
Uses of the Navier-Stokes equation range from semi-conductor engineering to astrophysics and it is not the goal of this section to attempt a review of these topics. Rather, I want to briefly describe an application where the interaction between academia and industry is perhaps not very well established. This concerns models of natural perils (NatCat models) that are used in the insurance industry to calculate risk capital requirements. These risk capital requirements are a determining factor for the solvency of a given company. In the EU the relevant regulatory framework is called Solvency~II (\cite{Level1}). There exist different NatCat models, which are in general proprietary, for natural disasters such as earthquakes, flooding, tropical cyclones, and extratropical cyclones (European winter storms). Storm models are often based on numerical weather prediction (NWP) systems, and therefore inherit all their advantages and flaws.  

\revise{However, it is not claimed that the stochastic HIPS formulation of this paper is appropriate to generate stochastic NatCat storm scenarios. Thus the material in this section  is  only intended as additional background information concerning possible applications of stochastic fluid mechanics, and it is logically independent from the main result, Theorem~\ref{thm:cNS}.
}

\subsubsection{Numerical weather prediction (NWP) and climate modeling}
NWP systems and climate models are important for a number of apparent reasons such as daily weather forecasts or climate change quantification.  

Geophysical flows are modeled by the compressible Navier-Stokes equations (without capillary term, i.e.\ $C=0$ in \eqref{intro:bNS1}). These equations are deterministic. However, any implementation of these equations needs to introduce temporal and spatial discretizations. Physical processes which  occur below these chosen grid scales (`subgrid phenomena') cannot be accounted for by any numerical model of the deterministic flow equations. Since the subgrid processes are inherently unknown and uncertain it seems reasonable to model these by a stochastic dynamics point of view. Indeed, such a position has been adopted quite early by Kraichnan~\cite{K68}. See also \cite{MR06} for a modern version and fluid equations with stochastic force terms. 

Reviews regarding NWP systems and climate models are contained in \cite{BTB15,Betal17,P19}. These also address the need for stochastic parameterizations of unknown subgrid processes. 

Recent advances in the stochastic modeling of geophysical flows include the `location uncertainty' approach of M\'emin, Resseguier and collaborators \cite{Mem14,RMC17,RMC17a,RMC17b} as-well as the `stochastic advection by Lie transport (SALT)' theory of Holm and collaborators \cite{Holm15,CGH17,CFH17,DHL20,ABHT20}. The SALT approach is based on the observation that subgrid phenomena represent unknown physical processes and should therefore be derived from a stochastic variational principle. As a consequence, these models preserve circulation along stochastic Lagrangian paths.

\subsubsection{NatCat models and solvency capital requirement (SCR)}
With the implementation of the Solvency~II regulatory regime (\cite{Level1}) per 1.\ January 2016, applied insurance mathematics has become a surprisingly diverse and multi-disciplinary subject. The relevant tools extend beyond classical actuarial science to, e.g., modeling of local general accounting principles (\cite{HG19,GHL20}), no-arbitrage principles and stochastic interest rate models (\cite{TW16,DHOST20}), and stochastic fluid dynamics. The significance of the last point in this (certainly incomplete) list is briefly explained below.   
  
One of the basic quantitative principles of Solvency~II can be summarized as follows: Insurance and reinsurance undertakings are required to quantify all relevant risk factors over a one-year horizon and derive a corresponding loss distribution. Now, the own funds (i.e.\ excess of assets over liabilities) have to cover the $99.5$ percentile of this distribution (`survival of the $200$ year event'). This $99.5$ percentile corresponds to the so-called solvency capital requirement (SCR). The SCR can be calculated from a prescribed standard formula or a company specific internal model. Medium and large sized companies generally use internal models. If a company has chosen the internal model approach and covers windstorm risks, then the corresponding loss distribution for a one-year period has to be derived. To do so the following three-step procedure, or a variant of it, is used (\cite{GK05}):
\begin{enumerate}
    \item 
    Hazard module: physical model. A large set of windstorm scenarios is generated. These are the so-called stochastic scenarios of the NatCat model.
    \item
    Vulnerability module. 
    This step quantifies the vulnerability (i.e., damage done) of the insured structures for each of the stochastic scenarios.   
    \item
    Financial loss module.
    Finally, structural damage is converted to loss by taking into account contract specifics (e.g., sum insured) for each damaged structure and, possibly, reinsurance.  
\end{enumerate}

Since the relevant time period for NatCat models is one year they operate on a new temporal scale when compared to short term weather forecast models and medium or long term climate models. The most advanced NatCat models are based on a coupling of a global circulation model (GCM) and a NWP system. Since these models are proprietary it is not possible to cite a suitable model documentation. However, \cite{C02} contains a review of the basic method, which is still quite up to date. In particular, it is explained how (deterministic) NWP systems are used to generate a set of scenario events by statistically sampling the initial conditions. 

\begin{quote} 
``Using NWP technology, a large set of potential future storms is generated
by taking data sets comprising the initial pressure fields of historical storms, perturbing them
both temporally and spatially, and moving them forward in time through the application of a
set of partial differential equations governing fluid flow. The resulting event set is rigorously
tested to ensure that it provides an appropriate representation of the entire spectrum of
potential storm experience -- not just events of average probability, but also the extreme
events that make up the tail of the loss distribution.''
(\cite{C02})
\end{quote}

This point will be taken up again in Section~\ref{sec:concl}. 

A recent reanalysis of European winterstorm events, which also highlights the potential financial risks for the insurance sector, has been carried out by \cite{Hetal15}.


\section{Notation and preliminaries}\label{sec:nota}

\subsection{Diffeomorphism groups}\label{sec:diffgps}
Let
$M=T^n=\R^n/\mathbb{Z}^n$. We fix $s>1+n/2$ and let $\textup{Diff}(M)^s$ denote the infinite dimensional $\cinf$-manifold of $H^s$-diffeomorphisms on $M$. 
Further, $\textup{Diff}(M)^s_0$ denotes the submanifold of volume preserving diffeomorphisms of Sobolev class $H^s$. Both, $\textup{Diff}(M)^s$ and $\textup{Diff}(M)^s_0$, are  topological groups but not Lie groups since left composition is only continuous but not smooth. Right composition is smooth.
The tangent space of $\textup{Diff}(M)^s$ (resp.\ $\textup{Diff}(M)^s_0$) at the identity $e$ shall be denoted by $\gu^s$ (resp.\ $\gu^s_0$). 
Let $\X^s(M)$ denote the vector fields on $M$ of class $H^s$ and $\X_0^s(M)$ denote the subspace of divergence free vector fields of class $H^s$.
We have $\gu^s_0 = \X^s_{0}(M)$ and $\gu^s=\X^s(M)$.
The superscript $s$ will be dropped from now on.  

We use right multiplication $R^g: \textup{Diff}(M)\to \textup{Diff}(M)$, $k\mapsto k\circ g = kg$ to trivialize the tangent bundle $T\textup{Diff}(M)\cong \textup{Diff}(M)\times\gu$, $\xi_g\mapsto(g,(TR^g)^{-1}\xi_g)$, and similarly for $\textup{Diff}(M)_0$. 

The Riemannian metric on $\textup{Diff}(M)\times\gu\cong T\textup{Diff}(M)$ is defined by
\[
  \ww<\xi_g,\eta_g> 
  = \int_M\vv<\xi(g(x)),\eta(g(x))>\, dx
\]
for $\xi,\eta\in\gu$, where $dx$ is the standard volume element in $M$, and $\vv<.,.>$ is the Euclidean inner product. 
See \cite{AK98,EM70,MEF,Michor06} for further background.

\subsection{Derivatives}
The adjoint with respect to $\ww<.,.>$ to the Lie derivative $L$, given by $L_X Y = \nabla_X Y - \nabla_Y X$, is 
\begin{equation}
    L^{\top}_X Y
    = 
    -\nabla_X Y - \textup{div}(X)Y - (\nabla^{\top}X)Y
\end{equation}
with
$
    \nabla_X Y 
    = \vv<X,\nabla>Y 
    = \sum X^i\del_i Y^j e_j 
$
and
$
    (\nabla^{\top}X)Y
    = \sum (\del_i X^j)Y^j e_i
$
with respect to the standard basis $e_i$, $i = 1,\dots,n$.
The notation $\ad(X)Y = [X,Y] = -L_X Y$ and $\ad(X)^{\bot} = -L^{\top}_X$ will be used. 
 
The variational derivative of a functional $F: \gu\to\R$ will be denoted by $\delta F/\delta X$, that is 
\begin{equation}
    \ww<\frac{\delta F}{\delta X},Y>
    = \frac{\del}{\del t}\Big|_0 F(X + tY)
\end{equation}
for $X,Y\in\gu$. 

\subsection{Semi-direct product structure}\label{sec:SD}
The configuration space of compressible fluid mechanics on $M$ is the semi-direct product $\textup{Diff}(M)\circledS\mathcal{F}(M)$ where $\mathcal{F}(M)$ denotes functions (also of Sobolev class $s$) on $M$. The semi-direct product structure is defined by the right action 
\begin{equation}
    R^{(\phi,g)}(\psi,f)
    = (\psi\circ\phi, f\circ\phi + g). 
\end{equation}
The corresponding phase space is trivialized with respect to this right multiplication as 
\[
 T^*(\textup{Diff}(M)\circledS\mathcal{F}(M)) \cong \textup{Diff}(M)\circledS\mathcal{F}(M) \times \gu^*\times\F(M)^*.
\]
The variables  $\mu\in\gu^*$ and $\rho\in\F(M)^*$ represent momentum and mass density, respectively. 
Making use of the Euclidean volume form $dx$, the duals can be identified as $\gu^*=\Om^1(M)$ and $\F(M)^*=\F(M)$. 
Details on Hamiltonian mechanics on semi-direct products are given in \cite{MRW84}, where also the case of compressible ideal fluids is treated.

\subsection{Stochastic dynamics}
Let $(\Om,\F,(\F_t)_{t\in[0,T]},P)$ be a filtered probability space satisfying the usual assumptions as specified in \cite{Pro}. In the following, all stochastic processes shall be understood to be adapted to this filtration.
The symbol 
\begin{equation}
    \ds 
\end{equation}
will be used to denote the Stratonovich differential to distinguish it from the variational derivative $\delta$. The Ito differential does not appear in this paper. The exterior differential is $d$.  

\subsection{Brownian motion in $\gu_0$}\label{sec:BMgu}
Let
\[
 \mathbb{Z}_n^+
 := \set{k\in\mathbb{Z}_n: k_1>0
 \textup{ or, for } i=2,\ldots,n,
 k_1=\ldots=k_{i-1}=0, k_i>0
 }.
\]
For $k\in\mathbb{Z}^+_n$ let $k_1^{\bot},\ldots,k_{n-1}^{\bot}$ denote a choice of pairwise orthogonal vectors in $\R^n$ such  that $|k_i^{\bot}|=|k|$ and $\vv<k_i^{\bot},k>=0$ for all $i=1,\ldots,n-1$. 

Consider the following system of vectors in $\gu_0$ (\cite{CM08,CS09}):
\begin{equation*}
    A_{(k,i)} =  \frac{1}{|k|^{s+1}}\cos\vv<k,x>k^{\bot}_i,\;
    B_{(k,i)} = \frac{1}{|k|^{s+1}}\sin\vv<k,x>k^{\bot}_i,\;
    A_{(0,j)} = e_j 
\end{equation*}
where $e_j\in\R^n$ is the standard basis and $s$ is the Sobolev index from Section~\ref{sec:diffgps}.  
By slight abuse of notation we identify these vectors with their corresponding right invariant vector fields on $\textup{Diff}(M)_0$. 

Further, in the context of the $\zeta_r$ vectors we shall make use of the multi-index notation $r = (k,i,a)$ where $k\in\mathbb{Z}_n^+$ and $a=0,1,2$ such that
\begin{align*}
    \zeta_r &= A_{(0,i)}
    \textup{ with } i=1,\ldots,n
    \textup{ if } a=0\\
    \zeta_r &= A_{(k,i)} 
        \textup{ with } i=1,\ldots,n-1
   \textup{ if } a=1\\
    \zeta_r &= B_{(k,i)} 
        \textup{ with } i=1,\ldots,n-1
    \textup{ if } a=2  
\end{align*}
Thus by a sum over $\zeta_r$ we shall mean a sum over these multi-indices, and this notation will be used throughout the rest of the paper. 

It can be shown (see \cite[Appendix]{CS09} for details) that the $\zeta_r$ form an orthogonal system of basis vectors in $\gu_0$, such that 
\begin{equation}\label{e:nablaXX} 
 \nabla_{\zeta_r}\zeta_r = 0
\end{equation} 
and, for $X\in\X(M)$,
\begin{equation}
\label{e:Delta}
 \sum \nabla_{\zeta_r}\nabla_{\zeta_r}X = c^s\Delta X
\end{equation}
where $c^s = 1+\frac{n-1}{n}\sum_{k\in\mathbb{Z}_n^+}\frac{1}{|k|^{2s}}$ is a constant and $\Delta$ is the vector Laplacian.

\begin{proposition}[\cite{CC07,CM08,C,DZ}]\label{prop:bm}
Let $W_t = \sum \zeta_r W_t^p$, where $W_t^r$ are independent copies of Brownian motion in $\R$. Then $W$ defines (a version of) Brownian motion (i.e., cylindrical Wiener process) in $\gu_0$. 
\end{proposition}

\section{HIPS for compressible flow}\label{sec:hips}

This section is concerned with the HIPS equations of motion and the mean field limit. Background on mean field theory can be found in \cite{S91,AD95,DV95,M96,JW17}. 

\subsection{The Hamiltonian}\label{sec:3A}
Consider a barotropic fluid in $M = \R^n/\mathbb{Z}^n$. At a (macroscopic) position $x\in M$ consider a tiny volume element $\Delta V_x$. Suppose $\Delta V_x$ is further divided into $N$ infinitesimal volume elements, $\Delta V_x^{\alpha}$, labeled by $\alpha = 1,\dots,N$, which are all assumed to have identical dimensions.   
In each volume element there is a mass density $\rho^{\alpha}=\rho^{\alpha}(x)$. Thus it is assumed that the continuum description of the fluid also holds at the level of the subdivsion. 
Let the fluid element in volume $\Delta V_x^{\alpha}$ have velocity $v^{\alpha}(x)$. 

Thus the energy of the particle ensemble in the total infinitesimal volume is determined by the velocities, $v^{\alpha}$, and densities, $\rho^{\alpha}$, in the subdivisions. To arrive at the total energy \eqref{e:H^N} we consider five contributions:
\begin{enumerate}
    \item 
    Translational kinetic energy;
    \item
    Equilibrium internal energy $\mathcal{U}$ which gives rise to the hydrostatic pressure $p$;
    \item
    Equilibrium capillary energy;
    \item
    Non-equilibrium expansion/compression energy;
    \item
    Stochastic energy due to molecular bombardment;
\end{enumerate}

\subsubsection{Translational kinetic energy}
The total velocity at $x$, $v(x)$, is the weighted average
\begin{equation}
\label{e:v}
    v(x)
    = 
    \frac{\sum_{\alpha=1}^N \rho^{\alpha}(x)v^{\alpha}(x) }{\sum_{\alpha=1}^N \rho^{\alpha}(x)}.
\end{equation}
The total momentum, $\mu(x)$, is therefore 
\begin{equation}
    \mu(x)
    = \sum_{\alpha}  \rho^{\alpha}(x)v^{\alpha}(x) 
    = \sum_{\alpha} \mu^{\alpha}
\end{equation}
and the translational kinetic energy $K^N(x)$ of the particle system is
\begin{equation}
\label{e:E1}
    K^N(x)
    = \frac{1}{2}\vv< \mu(x), v(x)> 
    = \frac{1}{2}
        \Big\langle 
        \sum_{\alpha} \mu^{\alpha} (x) , 
            \frac{ \sum_{\beta}\mu^{\beta}(x) }
                {\sum_{\nu}\rho^{\nu}(x)}
            \Big\rangle 
\end{equation}

\subsubsection{Equilibrium internal energy} 
Notice that the overall mass density in $\Delta V_x$ is expressed, in this subdivision picture, as 
$\rho(x) = \sum\rho^{\alpha}(x)/N$. Consequently, the mass contained in $\Delta V_x$ is 
$\rho(x)\cdot\Delta V_x 
= \sum\rho^{\alpha}\,dx$ where $\Delta V_x^{\alpha}\approx\Delta V_x/N$ is identified with the infinitesimal volume element $dx$.

The barotropicity assumption implies that the specific equilibrium internal energy in $\Delta V$ depends on the overall mass density $\sum\rho^{\alpha}/N$. 
Then the equilibrium internal energy in $\Delta V$ is
\begin{equation}
    \label{e:E2}
    \sum\rho^{\alpha}\mathcal{U}\Big(\sum\rho^{\alpha}/N\Big)\, dx. 
\end{equation}

\subsubsection{Capillary energy}
Let $\kappa\ge0$ be a  constant. 
The capillary energy is defined as 
\begin{equation}
    \label{e:E3}
    \kappa\by{1}{2}\Big\langle\nabla\sum\rho^{\alpha}, \nabla\sum\rho^{\alpha}/N\Big\rangle \,dx
\end{equation}
which depends, again, on the overall mass density. This form coincides with the capillary contribution to the Helmholtz free energy \cite[Equ.~(5)]{AMW97}. 

\subsubsection{Non-equilibrium expansion energy} 
Let $\nu\ge0$ be a constant. Assume temporarily that $\Delta V_x$ is a box which is aligned along Cartesian coordinates $e_1$, $e_2$, $e_3$ and that the flow has only velocity components pointing in the direction of the $e_1$-axis. Let $L$ denote the set of labels $\alpha$ such that the corresponding subvolumes $\Delta V^{\alpha}$ constitute the left wall of $\Delta V_x$ while $R$ denotes those which correspond to the right wall (viewed along $e_1$) of the volume. Now, if 
\begin{equation}
    \frac{\sum_{\alpha\in L}\rho^{\alpha} v^{\alpha}}{\sum_{\alpha\in L}\rho^{\alpha}}
    -
    \frac{\sum_{\alpha\in R}\rho^{\alpha} v^{\alpha}}{\sum_{\alpha\in R}\rho^{\alpha}}    
\end{equation}
is greater than $0$, then particles moving into $\Delta V$ are faster than those moving out, and the corresponding energy difference should contribute to the (non-equilibrium) internal energy of the system.  Since this expression is proportional to minus the divergence of the barycentric velocity, we propose a non-equilibrium internal energy expression 
\begin{equation}
    \label{e:E4}
    -\nu \sum\rho^{\alpha}\textup{div}\Big(\frac{\sum_{\alpha}\rho^{\alpha} v^{\alpha}}{\sum_{\alpha}\rho^{\alpha}} \Big)
    \revise{dx}.
\end{equation}

\subsubsection{Stochastic energy}
The energy due to molecular bombardment along a vector field $\xi_r$ \revise{in $\Delta V_x^{\alpha}\approx dx$} is
\begin{equation}
    \label{e:E5}
    \vv<\rho^{\alpha} v^{\alpha},\xi_r>\ds W^r \revise{dx}.
\end{equation}
See \cite{LCO08, HR12, Holm15, H17, H18}.
This stochastic perturbation corresponds to individual molecules imparting their velocities, namely $\xi_r$, on the macroscopic fluid element in $\Delta V^{\alpha}$. Since individual molecules are incompressible the $\xi_r$ are assumed to be divergence free.  

\subsubsection{Total energy}
\revise{
The above subdivision formulation implies that the total configuration space is of the form $\Pi_{\alpha=1}^N(\textup{Diff}(\Delta V^{\alpha})\circledS\F(\Delta V^{\alpha})))$.
But now the $\Delta V^{\alpha}$, which are all identical by assumption, are identified with the infinitesimal element $dx$ in $M$. Thus each index $\alpha$ corresponds to a copy of $M$, and this can be done since the for the state variables, velocity and density, it does not make a difference whether these are regarded at the $\Delta V$ or at the $\Delta V^{\alpha}$ level.
Therefore, letting the position $x\in M$ range over the full domain, the total configuration space is 
$\Pi_{\alpha=1}^N(\textup{Diff}(M)\circledS\F(M)) = (\textup{Diff}(M)\circledS\F(M))^N$. 
}

Let us switch from velocity and density to momentum \revise{(density)}, $\mu^{\alpha} = \rho^{\alpha}v^{\alpha}$, and density, $\rho^{\alpha}$, as state variables. 
\revise{The phase space is thus
$
 T^*(\textup{Diff}(M)\circledS\F(M))^N 
 = ( T^*(\textup{Diff}(M)\circledS\F(M)) )^N
$.
}
We use the Euclidean metric to identify \revise{each copy in} the phase space, which is the regular dual, as 
\revise{
\[
 T^*(\textup{Diff}(M)\circledS\F(M)) 
 = T(\textup{Diff}(M)\circledS\F(M))
 = (\textup{Diff}(M)\circledS\F(M))\times(\X(M)\circledS\mathcal{F}(M))
\]
where the last identification follows from right-multiplication in the semi-direct product group, see Section~\ref{sec:SD}.
}
The resulting Hamiltonian of the IPS will therefore be a function
\begin{equation}
\label{e:HSys}
    H^N:
    \Big( T(\textup{Diff}(M)\circledS\F(M)) \Big)^N \to \R.
\end{equation}
For 
\begin{equation}
\label{e:phase}
    \Gamma
    =
    \Big( \Phi^{\alpha},f^{\alpha}; \mu^{\alpha},\rho^{\alpha}\Big)_{\alpha=1}^N
    \in \Big( T(\textup{Diff}(M)\circledS\F(M)) \Big)^N 
\end{equation}
the total Hamiltonian is the sum of \eqref{e:E1}, \eqref{e:E2}, \eqref{e:E3}, \eqref{e:E4} and \eqref{e:E5}, and given 
in semi-martingale notation as
\begin{align}
    \label{e:H^N}
    H^N(\Gamma)
    &=
    \frac{1}{2}\int_M \Big\langle 
        \sum_{\alpha} \mu^{\alpha} , 
            \frac{ \sum_{\beta}\mu^{\beta} }
                {\sum_{\gamma}\rho^{\gamma}}
            \Big\rangle \,dx \ds t
    \\
\notag 
    &\phantom{==}
    +  \int_M \sum_{\alpha}\rho^{\alpha} \mathcal{U}\Big(\frac{\sum_{\beta}\rho^{\beta}}{N}\Big)\,dx
    \ds t \\
\notag 
    &\phantom{==}
    + \kappa\int_M
    \Big\langle \sum_{\alpha}\nabla\rho^{\alpha}, \sum_{\alpha}\nabla\rho^{\alpha}/N \Big\rangle \,dx \ds t
\\
\notag 
    &\phantom{==}
    - \nu \int_M \sum_{\alpha}\rho^{\alpha}                  \textup{div}\Big(\frac{\sum_{\beta}\mu^{\beta}}{\sum_{\gamma}\rho^{\gamma}}\Big) 
    \,dx
    \ds t
    \\
\notag 
    &\phantom{==}
    + \eps\int_{M}\sum_{j,\alpha}\vv<\mu^{\alpha},\xi_r>\,dx \ds W^{r,\alpha}
\end{align}
where $W^{r,\alpha}$ are pairwise independent Brownian motions such that $[W^r,W^s]_t=\delta_{r,s}t$, the solenoidal vector fields $\xi_r\in\X_0(M)$ are fixed and $\eps\ge 0$ is a constant.  
The Hamiltonian is right-invariant by construction as it does not depend on $(\Phi^{\alpha},f^{\alpha})\in \textup{Diff}(M)\circledS\F(M)$.

\subsection{HIPS equations of motion}\label{sec:3b}
\revise{
The phase space \eqref{e:HSys} is an $N$-fold direct product of a tangent bundle identified with its dual. It carries therefore the corresponding direct product canonical symplectic form.  Since the Hamiltonian $H^N$ does not depend on $(\Phi^{\alpha},f^{\alpha})_{\alpha=1}^N\in (\textup{Diff}(M)\circledS\F(M))^N$
we can pass via Lie-Poisson reduction to the phase space $( \X(M)\circledS\mathcal{F}(M) )^N$.

The Hamiltonian IPS equations of motion follow therefore from the variational derivatives (again, in the semi-martingale notation)
}
\begin{align}
\label{e:hips1}
    \frac{\delta H^N}{\delta\mu^{\alpha}}
    &= 
    \frac{ \sum_{\beta}\mu^{\beta} }{ \sum_{\gamma}\rho^{\gamma} }\ds t
    + \nu\nabla\log\sum_{\beta}\rho^{\beta}\ds t 
    + \eps\sum_r \xi_r\ds W^{r,\alpha} \\
\notag
    &= 
    u^N\ds t
    + \nu\nabla\log\rho^N \ds t 
    + \eps \sum_r \xi_r\ds W^{r,\alpha} \\
\label{e:hips2}
    \frac{\delta H^N}{\delta\rho^{\alpha}} 
    &= 
    \Big(- \frac{1}{2}\Big\langle 
        \frac{\sum_{\beta} \mu^{\beta}}{{\sum_{\gamma}\rho^{\gamma}}} , 
            \frac{ \sum_{\beta}\mu^{\beta} }{\sum_{\gamma}\rho^{\gamma}}
            \Big\rangle
            + \mathcal{U}\Big(\frac{\sum_{\beta}\rho^{\beta}}{N}\Big) 
            + \frac{\sum_{\beta}\rho^{\beta}}{N}\mathcal{U}'\Big(\frac{\sum_{\beta}\rho^{\beta}}{N}\Big) \\
\notag 
    &\phantom{==} 
    - \kappa \Delta \frac{\sum_{\beta}\rho^{\beta}}{N}
    - \nu \frac{1}{\sum_{\beta}\rho^{\beta}}\textup{div}\Big(\sum_{\beta}\mu^{\beta}\Big)
        \Big)\ds t\\
\notag 
    &= 
    \Big(- \frac{1}{2}\vv<u^N,u^N>
            + \mathcal{U}(\rho^N) 
            + \rho^N\mathcal{U}'(\rho^N) 
            - \kappa\Delta \rho^N
            - \nu \frac{1}{\rho^N}\textup{div}(\mu^N)
        \Big)\ds t
\end{align}
\revise{
by using the $N$-fold product of the semi-direct product structure (\cite{MRW84}). Here
}
the  abbreviations
$\mu^N  
:= \frac{\sum_{\beta}\mu^{\beta}}{N}$,
$\rho^N
:= \frac{\sum_{\beta}\rho^{\beta}}{N}$
and
$u^N
:= \frac{\mu^N}{\rho^N}$
are used.

\revise{
\begin{remark}
The variational derivatives, $\delta H^N/\delta\mu^{\alpha}$ and $\delta H^N/\delta\rho^{\alpha}$, also depend on $\mu^{\beta}$ and $\rho^{\beta}$ with $\beta\neq\alpha$. Hence the right hand sides in the equations~\eqref{e:hips3}-\eqref{e:hips5} below cannot be viewed as vector fields on $\X(M)\circledS\mathcal{F}(M)$, but only as Cartesian projections of vector fields on the full space $( \X(M)\circledS\mathcal{F}(M) )^N$.
\end{remark}
}

\begin{remark}
Equations~\eqref{e:hips1} and \eqref{e:hips2} depend only on the empirical averages $\mu^N$ and $\rho^N$. 
\end{remark}

\begin{remark}
The quantity $u^N$ is the specific momentum (viewed as a vector field) of the ensemble average. But it is (for non-constant density) not equal to the empirical average of velocities, that is 
$
    u^N
    \neq 
    \frac{\sum \mu^{\alpha}/\rho^{\alpha}}{N}
$.
\end{remark}
The stochastic Hamilton equation associated to \eqref{e:H^N} for 
$(\Phi^{\alpha}, \mu^{\alpha},\rho^{\alpha})$ are:
\begin{align}
\label{e:hips3}
    \ds \Phi^{\alpha}_t
    &=
    \Big(\frac{\delta H^N}{\delta\mu^{\alpha}}\Big)\circ\Phi_t^{\alpha} \\
\notag
    &= 
    \Big(u^N
    + \nu\nabla\log\rho^N\Big)\circ\Phi_t^{\alpha} \ds t 
    + \eps\sum_r \xi_r\circ\Phi_t^{\alpha} \ds W^{r,\alpha} \\
\label{e:hips4}
    \ds\mu^{\alpha}_t
    &=
    - \ad\Big(
        (\delta_t\Phi_t^{\alpha})\circ(\Phi_t^{\alpha})^{-1}
        \Big)^{\top}
        \mu^{\alpha}_t
    -  \frac{\delta H^N}{\delta\rho^{\alpha}} \diamond \rho^{\alpha}_t \\
\notag 
    &=
    -\nabla_{\delta H^N/\delta\mu^{\alpha}}\mu^{\alpha} 
    - \textup{div}(\delta H^N/\delta\mu^{\alpha})\mu^{\alpha} 
    - (\nabla^{\top}\delta H^N/\delta\mu^{\alpha})\mu^{\alpha} 
    -\rho^{\alpha}\nabla \frac{\delta H^N}{\delta\rho^{\alpha}} \\
\notag 
    &=
    \Big(-\nabla_{u^N + \nu\nabla\log\rho^N}\mu^{\alpha} - \textup{div}(u^N+\nu\nabla\log\rho^N)\mu^{\alpha} - (\nabla^{\top}u^N + \nu\nabla^{\top}\nabla\log\rho^N)\mu^{\alpha} \\
\notag 
    &\phantom{==}    
    -\rho^{\alpha}
        \Big(
        -(\nabla^{\top}u^N)u^N 
        + (\rho^N)^{-1} \nabla\Big( (\rho^N)^2\mathcal{U}'(\rho^N) \Big)
        - \kappa\nabla\Delta\rho^N
    \\
\notag 
    &\phantom{====}
        +\nu(\rho^N)^{-2}\textup{div}(\mu^N)\nabla\rho^N
        -\nu(\rho^N)^{-1}\nabla\textup{div}(\mu^N)
        \Big)
    \Big)\ds t\\
\notag 
    &\phantom{==}
    - \eps\sum_r \ad\Big(\xi_r\Big)^{\top}\mu^{\alpha}\ds W^{r,\alpha}  
    \\
\label{e:hips5}
    \ds\rho_t
    &= 
     - L_{(\delta_t\Phi_t^{\alpha})\circ(\Phi_t^{\alpha})^{-1} }\rho^{\alpha}_t \\
\notag
    &= -\textup{div}\Big(
        \rho^{\alpha} u^N + \nu\rho^{\alpha}\nabla\log\rho^N
        \Big)\ds t
       - \eps\sum_r \textup{div}\Big(\rho^{\alpha} \xi_r\Big) \ds W^{r,\alpha}
\end{align}
Here $\rho^{\alpha}$ is viewed as a density whence $L$ is the Lie derivative of a density, not of a function.
The momentum variable is identified, via the Euclidean metric, as an element
\begin{equation}
    \mu^{\alpha}\in\X(M)
\end{equation}
whence the transpose Lie derivative $L^{\top}$ is used instead of $L^*$.  
The diamond notation in \eqref{e:hips4} is defined by 
$f\diamond\rho = \rho\nabla f$ and this term arises because of the semi-direct product structure.

\subsection{Mean field limit}\label{sec:3c}
Assume the mean field limit of \eqref{e:hips3}-\eqref{e:hips5} exists, for $N\to\infty$. Since all subvolumes $\Delta V^{\alpha}$ and their enclosed fluid elements are identical it suffices to consider $\alpha=1$,
\begin{equation}
\label{e:mflimit}
    (\Phi_t^1,f_t^1,\mu_t^1,\rho_t^1)\longto (\Phi_t,f_t,\mu_t,\rho_t)
\end{equation}
as $N\to\infty$. 
Hence
\begin{equation}
    \mu^N \longto E[\mu] =: \bar{\mu} 
    \quad\textup{and}\quad
    \rho^N\longto E[\rho] =: \bar{\rho}  
\end{equation}
and
\begin{equation}
\label{e:mf1}
    u^N
    =
    \frac{ \mu^N }{ \rho^N }
    \longto 
    \bar{\mu}/\bar{\rho} 
    =: u
\end{equation}
as $N\to\infty$.

Note that $\mu$ and $\rho$ are stochastic processes while $u$ is deterministic.
The mean field limit equations of motion for $\mu$ and $\rho$ are
\begin{align}
\label{e:Hmf1}
    \ds\mu 
    &=
    \Big(
        -\nabla_{u+\nu\nabla\log\bar{\rho}}\mu 
        - \textup{div}(u+\nu\nabla\log\bar{\rho})\mu 
        - (\nabla^{\top}u)\mu
        - \nu(\nabla^{\top}\nabla\log\bar{\rho})\mu 
    \\
\notag
    &\phantom{==}
        + \rho (\nabla^{\top}u) u
        - \rho \bar{\rho}^{-1} \nabla\Big( \bar{\rho}^2\mathcal{U}'(\bar{\rho}) 
    \Big) 
    + \kappa\rho\nabla\Delta\bar{\rho}
\\
\notag
    &\phantom{==}
    - \nu\rho \bar{\rho}^{-2}\textup{div}(\bar{\mu})\nabla\bar{\rho}
    + \nu\rho \bar{\rho}^{-1}\nabla\textup{div}(\bar{\mu})
    \Big)  \ds t
 \\
\notag
    &\phantom{==}
    - \eps\sum_r \ad(\xi_r)^{\top}\mu\ds W^r  
    \\
\label{e:Hmf2}
    \ds\rho_t
    &= -\textup{div}(\rho u + \nu\rho\nabla\log\bar{\rho})\ds t
       - \eps\sum_r \textup{div}(\rho \xi_r) \ds W^r
\end{align}

These equations are linear in $\mu$ and $\rho$, and depend otherwise on the mean fields $E[\mu] = \bar{\mu}$ and $E[\rho] = \bar{\rho}$. 
Let $p := \bar{\rho}^2\mathcal{U}'(\bar{\rho})$. Using that $\bar{\mu} = \bar{\rho}u$, the equations for the expectations $\bar{\mu}$, $\bar{\rho}$ are therefore
\begin{align}
\label{e:exp1}
    \dot{\bar{\mu}}
    &=
    -\nabla_{u+\nu\nabla\log\bar{\rho}}\bar{\mu} 
    - \textup{div}(u+\nu\nabla\log\bar{\rho})\bar{\mu}
    - \nu(\nabla^{\top}\nabla\log\bar{\rho})\bar{\mu}
    \\
\notag
&\phantom{==}
    - \nabla p
    + \kappa\bar{\rho}\nabla\Delta\bar{\rho}
    - \nu\bar{\rho}^{-1}\textup{div}(\bar{\mu})\nabla\bar{\rho}
    + \nu\nabla\textup{div}(\bar{\mu}) 
    + \frac{\eps^2}{2}\sum_r L^{\top}_{\xi_r} L^{\top}_{\xi_r}\bar{\mu} 
    \\
    \label{e:exp2}
    \dot{\bar{\rho}}
    &= -\textup{div}\Big(\bar{\rho} u + \nu\nabla\bar{\rho}\Big) 
       + \frac{\eps^2}{2}\sum_r \textup{div}\Big( \textup{div}(\bar{\rho} \xi_r) \xi_r \Big)   
\end{align}

\subsection{Barotropic Navier-Stokes equation}\label{sec:3d}
Assume that the perturbation vector fields $\xi_r$ are given by $\zeta_r$, defined in Section~\ref{sec:BMgu} and that $\eps^2 c^s/2 = \nu$. 
Then \eqref{e:Delta} implies 
\begin{equation}
 \frac{\eps^2}{2}\sum_r L^{\top}_{\zeta_r} L^{\top}_{\zeta_r}\bar{\mu} = \nu\Delta\bar{\mu} 
 \;\textup{ and }\;
 \frac{\eps^2}{2}\sum_r \textup{div}\Big( \textup{div}(\bar{\rho} \zeta_r) \zeta_r \Big)  
 = \nu \Delta\bar{\rho}. 
\end{equation}
(The explicit calculation is carried out in \cite[Lemma~4.3]{H18}.)

Hence equations~\eqref{e:exp1} and \eqref{e:exp2} become
\begin{align} 
\label{e:exp11}
    \dot{\bar{\mu}}
    &=
    -\nabla_{u+\nu\nabla\log\bar{\rho}}\bar{\mu} 
    - \textup{div}(u+\nu\nabla\log\bar{\rho})\bar{\mu} 
    - \nu(\nabla^{\top}\nabla\log\bar{\rho})\bar{\mu}\\
\notag
&\phantom{==}
    - \nabla p
    + \kappa\bar{\rho}\nabla\Delta\bar{\rho}
    - \nu\bar{\rho}^{-1}\textup{div}(\bar{\mu})\nabla\bar{\rho}
    + \nu\nabla\textup{div}(\bar{\mu}) 
    + \nu\Delta\bar{\mu} \\
\label{e:exp22}
    \dot{\bar{\rho}}
    &= -\textup{div}(\bar{\rho} u)
\end{align}
Therefore,
\begin{align*}
    \dot{\bar{\mu}}
    &= \dd{t}{}(\bar{\rho}u)
    = -\textup{div}(\bar{\rho}u)u + \bar{\rho}\dot{u}\\
    &=
    -\bar{\rho}\nabla_{u+\nu\nabla\log\bar{\rho}}u
    - \textup{div}(\bar{\rho}u)u 
    - \nu\textup{div}(\bar{\rho}\nabla\log\bar{\rho})u \\
    &\phantom{==}
    - \nu(\nabla^{\top}\nabla\log\bar{\rho})\bar{\rho}u 
    - \nabla p
    + \kappa\bar{\rho}\nabla\Delta\bar{\rho}
    - \nu\bar{\rho}^{-1}\vv<\nabla\bar{\rho},u>\nabla\bar{\rho} 
    - \nu\textup{div}(u)\nabla\bar{\rho} \\
    &\phantom{==}
    + \nu\nabla\Big(\vv<\nabla\bar{\rho},u> + \bar{\rho}u\Big)
    + \nu\Delta(\bar{\rho}u) \\
    &=
    -\bar{\rho}\nabla_{u+\nu\nabla\log\bar{\rho}}u
    - \textup{div}(\bar{\rho}u)u 
    - \nu(\Delta\bar{\rho})u \\
    &\phantom{==}
    - \nu\nabla_u\nabla\bar{\rho} 
    + \nu\bar{\rho}^{-1}\vv<u,\nabla\bar{\rho}>\nabla\bar{\rho} 
    - \nabla p 
    + \kappa\bar{\rho}\nabla\Delta\bar{\rho}\\
    &\phantom{==}
    - \nu\bar{\rho}^{-1}\vv<\nabla\bar{\rho},u>\nabla\bar{\rho} 
    - \nu\textup{div}(u)\nabla\bar{\rho} 
    + \nu(\nabla^{\top}\nabla\bar{\rho}) u 
    + \nu (\nabla^{\top}u) \nabla\bar{\rho} \\
    &\phantom{==}
    + \nu \textup{div}(u)\nabla\bar{\rho} 
    + \nu \bar{\rho}\nabla\textup{div}(u) 
    + \nu(\Delta\bar{\rho})u\\
    &\phantom{==}
    + 2\nu\nabla_{\nabla\bar{\rho}}u
    + \nu\bar{\rho}\Delta u \\
    &=
    - \bar{\rho}\nabla_{u}u
    - \textup{div}(\bar{\rho}u)u 
    - \nabla p 
    + \kappa\bar{\rho}\nabla\Delta\bar{\rho}
    + \nu (\nabla^{\top}u)\nabla\bar{\rho} 
    + \nu \bar{\rho}\nabla\textup{div}(u)
    + \nu \nabla_{ \nabla\bar{\rho} }u
    + \nu\bar{\rho}\Delta u
\end{align*}
Define the stress tensor, $S$, by 
\begin{equation}
    \label{e:stress}
    S_{ij}
    = \nu\bar{\rho}\Big(\del_i u^j + \del_j u^i\Big)
\end{equation}
and the corresponding force  
\begin{align}
    \textup{div}\,S
    = \sum \del_i S_{ij} e_j 
    = 
    \nu (\nabla^{\top}u)\nabla\bar{\rho} 
    + \nu \bar{\rho}\nabla\textup{div}(u)
    + \nu \nabla_{ \nabla\bar{\rho} }u
    + \nu\bar{\rho}\Delta u.
\end{align}
The capillary tensor, $C$, is defined (see \cite[Equ.~(4)]{AMW97}) by 
\begin{equation}
    \label{e:Cap}
    C 
    = \kappa\Big(
        \Big( \bar{\rho}\Delta\bar{\rho}
            +\by{1}{2}\vv<\nabla\bar{\rho},\nabla\bar{\rho}> \Big)\mathbb{I}
        - \nabla\bar{\rho}\otimes\nabla\bar{\rho} \Big)
\end{equation}
and satisfies $\textup{div}\,C = \kappa\bar{\rho}\nabla\Delta\bar{\rho}$.

Note that $\nu\ge0$ and $\kappa\ge0$ are constants, and that $\bar{\mu}$, resp.\ $\bar{\rho}$ are a time dependent vector field, resp. function by construction.
It follows that: 

\begin{theorem}\label{thm:cNS}
The mean field equations \eqref{e:Hmf1} and \eqref{e:Hmf2} imply, if $\xi_r=\zeta_r$ and $c^s\eps^2/2=\nu$, that the expectations $\bar{\mu}=E[\mu]$ and $\bar{\rho}=E[\rho]$ satisfy the compressible Navier-Stokes equations
\begin{align}
\label{e:bNS1}
    \dot{u}
    &=   
    - \nabla_{u}u
    - \bar{\rho}^{-1}\nabla p
    + \bar{\rho}^{-1} \Big(\textup{div}\,S  + \textup{div}\,C\Big) \\
\label{e:bNS2}
    \dot{\bar{\rho}}
    &= 
    -\textup{div}(\bar{\rho}u)
\end{align}
where $u=\bar{\mu}/\bar{\rho}$. 
\end{theorem}

\section{Stochastic Kelvin Circulation Theorem}\label{sec:KCT}
\revise{
In \cite{DH19} it is shown that stochastic Euler-Poincar\'e fluid equations are characterized by preserving circulation along Lagrangian paths. 
}
Since \eqref{e:Hmf1}-\eqref{e:Hmf2} are obtained as a mean field limit of a Hamiltonian IPS, 
\revise{
and can be viewed of as mean field generalization of the stochastic fluid system in \cite{DH19}, 
}
there should be a Kelvin Circulation Theorem:

\begin{proposition}\label{prop:KCT}
Let $C$ be a smooth closed loop which is transported by the Lagrangian flow $\Phi_t$, defined through the mean field limit \eqref{e:mflimit} and characterized by 
\begin{equation}
    \ds\Phi_t\circ\Phi_t^{-1}
    = \Big(u_t+\nu\nabla\log\bar{\rho}\Big)\ds t + \eps\sum \xi_r\ds W^r_t.
\end{equation}
Let $\mu_t$ and $\rho_t$ be solutions of \eqref{e:Hmf1} and \eqref{e:Hmf2}. Then
\begin{equation}\label{e:KCT}
    \ds\int_{(\Phi_t)_*C}\rho_t^{-1}\mu_t^{\flat} = 0    
\end{equation}
where $\flat$ is the Euclidean isomorphism to one-forms (since $\mu$ is treated as a vector field). 
\end{proposition}

\begin{proof}
Equations~\eqref{e:hips3} and \eqref{e:mflimit} yield
\begin{equation}
    \label{e:Phi}
    \ds \Phi
    = \Big(u+\nu\nabla\log\bar{\rho}\Big)\circ\Phi\ds t
      + \eps\sum\xi_r\circ\Phi \ds W^r
\end{equation}

Now, \eqref{e:Hmf1} and \eqref{e:Hmf2} imply that $X_t := \rho_t^{-1}\mu_t$ satisfies
\begin{align}
    \ds X
    &=
    -\rho^{-2}(\ds\rho)\mu
        + \rho^{-1}\ds\mu 
\label{e:Xvel} \\
    &=
    \Big(
        - \nabla_{u+\nu\nabla\log\bar{\rho}}X 
        - (\nabla^{\top} u)X 
        - \nu(\nabla^{\top}\nabla\log\bar{\rho})X
        - \nabla\tilde{p}
    \Big)\ds t 
\notag \\
    &\phantom{==}
    -
    \eps\sum \Big(
        -\nabla_{\xi_r} X 
        - (\nabla^{\top} \xi_r)X 
    \Big)\ds W^r
\notag
\end{align}
with
\begin{equation}
    \tilde{p}
    :=
    -\frac{1}{2}\vv<u,u>
    +\mathcal{U}(\bar{\rho}) 
    + \bar{\rho}\mathcal{U}'(\bar{\rho})
    - \kappa\Delta\bar{\rho}
    - \nu\bar{\rho}^{-1}\textup{div}(\bar{\mu}).
\end{equation}
Hence, with parameterization $C = c([0,1])$:  
\begin{align*}
    \ds \int_{(\Phi_t)_*C} X_t^{\flat}
    &= \ds\int_C \Phi_t^* X_t^{\flat} 
    = \ds\int_C \Big( X_t^{\flat} \circ \Phi_t\Big). T\Phi_t \\
    &= \int_0^1\ds\vv< X_t\circ\Phi_t, T\Phi_t.c'(s)>\,ds \\
    &=
    \int_0^1\Big(
        \vv<(\ds X_t)\circ\Phi_t + TX_t.\delta\Phi_t, T\Phi_t.c'(s)> \\
    &\phantom{===}
        +
        \vv<X_t\circ\Phi_t, 
            T(u_t\ds t 
            + \nu\nabla\log\bar{\rho} \ds t 
            + \eps\sum\xi_r\ds W^r_t ).T\Phi_t.c'(s)>
        \Big)\,ds\\
    &= 
    \int_0^1\Big(
        \Big\langle(\ds X_t)\circ\Phi_t 
            + (\nabla_{u_t\,\delta_t t 
                + \nu\nabla\log\bar{\rho} \,\delta_t t
                + \eps\sum\xi_r\,\delta_t W^r_t } X_t) \circ\Phi_t\\
    &\phantom{===}
            + ((\nabla^{\top}u_t\ds t 
                + \nu\nabla^{\top}\nabla\log\bar{\rho} \ds t
                + \eps\sum\nabla^{\top}\xi_r\ds W^r_t )X_t)\circ\Phi_t 
        , T\Phi_t.c'(s)\Big\rangle
        \Big)\,ds\\
    &=
    \int_{(\Phi_t)_*C}
        \Big(
         \ds X_t
            + \nabla_{u_t\,\delta_t t 
                + \nu\nabla\log\bar{\rho} \,\delta_t t
                + \eps\sum\xi_r\,\delta_t W^r_t } X_t
                \\
    &\phantom{===}
            + (\nabla^{\top}u_t\,\delta_t t 
                + \nu\nabla\log\bar{\rho} \ds t
                + \eps\sum\nabla^{\top}\xi_r\,\delta_t W^r_t ) X_t
        \Big)^{\flat}
        \\
    &= 0 
\end{align*}
since, by \eqref{e:Xvel}, the integrand equals $-(\nabla\tilde{p})^{\flat}\ds t = -d\tilde{p}\ds t$. 
\end{proof}

\section{Conclusions}\label{sec:concl}

\subsection{HIPS approach to the compressible Navier-Stokes equation}
The mean field system \eqref{e:Hmf1}-\eqref{e:Hmf2} is derived from the interacting particle point of view under the basic assumption that the equations of motion follow from stochastic Hamiltonian mechanics. 
Therefore, circulation is preserved along stochastic Lagrangian paths.
If the perturbation fields $\xi_r$ run over the orthogonal system $\zeta_r$, defined in Section~\ref{sec:BMgu}, such that the stochastic perturbation is given by a cylindrical Wiener process, then the mean fields $E[\mu]$ and $E[\rho]$ solve the compressible Navier-Stokes equations.  (Theorem~\ref{thm:cNS} and Proposition~\ref{prop:KCT}.)


While the HIPS formulation relies on a system of $N$ interacting SDEs, the mean field equations \eqref{e:Hmf1}-\eqref{e:Hmf2} is a single SDE system for momentum and mass density. In contrast to statistical mechanics, this mean field formulation is obtained without any closure assumptions. (However, in this paper the existence of the mean field limit is not proved but assumed.) 

In the mean field limit \eqref{e:mflimit}, the HIPS evolution equation \eqref{e:hips3} becomes
\begin{equation}
\label{e:mf_transp}
    (\ds \Phi_t)\circ\Phi_t^{-1}
    = \Big( E[\mu_t]/E[\rho_t]+\nu\nabla\log E[\rho_t] \Big)\ds t
      + \eps \sum \xi_r\ds W_t^r
\end{equation}
which is of similar form as the LA SALT advection field 
\begin{equation}
\label{e:la_salt}
    E[u_t^L]\ds t + \sum  \xi_r\ds W_t^r
\end{equation}
of \cite{DHL20,ABHT20}, where $u_t^L$ is a stochastic velocity field. 
However, there are a few crucial differences: while \eqref{e:la_salt} is the starting point for LA SALT theory, the HIPS formulation is based on the ensemble Hamiltonian \eqref{e:H^N} and the Lie transport along \eqref{e:mf_transp} in the mean field equations \eqref{e:Hmf1}-\eqref{e:Hmf2} is a consequence of the Hamiltonian structure of the IPS (and the ensuing passage to the mean field limit). Moreover, unless the density is constant, it is not clear how to identify the drift in \eqref{e:mf_transp} with the expectation of a velocity.   
Thus, both, the starting points and the advection fields are different. However, the perturbation fields $\xi_r$ can be interpreted in the same manner.


\subsection{NatCat modeling of windstorm events}
\revise{NatCat models used to calculate the solvency capital requirement (SCR as defined in \cite{Level1}) for storm risks rely on NWP systems. 
These NWP systems are deterministic and, to arrive at a set of `stochastic' NatCat scenarios, the initial conditions are statistically sampled. 
As discussed in the Introduction, all such numerical schemes suffer from subgrid phenomena, and for geophysical flow models a well-established means for treating these deficiencies is by stochastic fluid mechanics (\cite{ABHT20,Betal17,Holm15,Mem14,RMC17}).
Since SCR calculation is concerned with predicting extreme events in the $99.5$ percentile, and not only average storm patters, it seems reasonable to expect that also NatCat models would benefit from a stochastic dynamics approach.  

However, it is not claimed that the stochastic HIPS formulation of this paper is appropriate to generate stochastic NatCat storm scenarios. 
}

\end{document}